\documentclass[twocolumn]{autart}
\PassOptionsToPackage{dvipsnames}{xcolor}
\makeatletter
\renewcommand\@seccntformat[1]{\csname lb@#1\endcsname}
\makeatother
\usepackage{graphicx}
\usepackage{amsmath,amssymb}
\usepackage{booktabs}
\usepackage{tikz}
\usetikzlibrary{decorations.pathreplacing,fit,calc}
\pgfdeclarelayer{background}
\pgfsetlayers{background,main}
\usepackage{colortbl}
\usepackage{algorithmic}
\usepackage[normalem]{ulem}
\makeatletter\let\AND\@undefined
\def\algorithm{\medskip\par\noindent\minipage{\linewidth}%
  \let\@makecaption\@makealgocaption\def\@captype{algorithm}\small}

\makeatother
\usepackage[numbers,sort&compress]{natbib}
\usepackage[colorlinks=true, linkcolor=blue, citecolor=blue, hypertexnames=false]{hyperref}

\newcommand{\IP}{\mathit{IP}}
\newcommand{\SL}{\mathit{SL}}
\newcommand{\Nc}{\mathcal{N}}
\newcommand{\Ec}{\mathcal{E}}

\newcommand{\Sc}{\mathcal{S}}
\newcommand{\Cc}{\mathcal{C}}
\newcommand{\im}{\text{\normalfont\j}}
\newcommand{\pd}[2]{\frac{\partial #1}{\partial #2}}

\newcommand{\Ppaths}[2]{\mathcal{P}_{#1 \to #2}}
\newcommand{\Tp}[1]{\mathcal{T}_{#1}}
\newcommand{\pe}[1]{P^{#1}_{e}}

\newcommand{\squeeze}[1]{%
  \resizebox{\ifdim\width>\linewidth \linewidth\else\width\fi}{!}{%
    \ensuremath{\displaystyle #1}}}

\usepackage{xcolor}
\usepackage[deletedmarkup=sout,authormarkup=superscript]{changes}
\definechangesauthor[name={Lauren}, color=NavyBlue]{LB}
\definechangesauthor[name={Dominic}, color=RedViolet]{DLM}
\definechangesauthor[name={Adam}, color=OliveGreen]{AC}

\AtBeginDocument{\setlength{\abovedisplayskip}{5pt plus 1pt minus 2pt}\setlength{\belowdisplayskip}{5pt plus 1pt minus 2pt}\setlength{\abovedisplayshortskip}{2pt}\setlength{\belowdisplayshortskip}{3pt}}
\begin{document}

\begin{frontmatter}

\title{Network Design against the Bullwhip Effect in Complex Supply
Chains\thanksref{footnoteinfo}}

\thanks[footnoteinfo]{This work was supported by the Natural Sciences and Engineering Research Council of Canada Discovery Grants Program (Reference \#\,RGPIN-2023-03257) and the Wall Legacy Awards. Corresponding author: D. Liao-McPherson.}

\author[addr]{Lauren Bogo}\ead{lauren.bogo@ubc.ca}~and
\author[addr]{Dominic Liao-McPherson}\ead{dliaomcp@mech.ubc.ca}

\address[addr]{Mechanical Engineering, The University of British Columbia, Vancouver, Canada}

\begin{abstract}
This paper studies the bullwhip effect, the amplification of demand
fluctuations into larger order fluctuations upstream, in supply chain networks where
each firm orders from several suppliers over routes with different lead times.
We show that on a directed acyclic network every response from demand to
orders is a sum over paths of products of nodal responses, so that the
stability and stability margins of the network are determined by those of its
individual nodes, and that a node which splits its orders across routes of
different lengths gains margin it can spend on a higher gain and a faster
response. We then minimize the worst-case amplification of the network by
choosing how each firm splits its orders among its suppliers, which route lead
times to shorten, and each firm's ordering gain, with gradients from one
adjoint pass per demand node and frequency. On a 14-city network the optimized design amplifies 20~dB
less than cost-minimal routing and 17~dB less than a design that shortens
the longest routes with the same budget, and is better connected than either.
Amplification is therefore a property of routing and delay, not only of the
ordering policy, and designing against it also makes the network more robust
to a lost route.
\end{abstract}

\end{frontmatter}

\section{Introduction}
\label{sec:intro}

In modern market economies, supply chains are a complex and dynamic interconnected network of suppliers, manufacturers, retailers, wholesalers, and consumers that transform raw materials into finished products and deliver these products to consumers. A common issue in supply chains is the \emph{bullwhip effect} where fluctuations in demand for a product are amplified upstream in the network \citep{lee1997distortion,wang2016progress}. These oscillations drive up costs and can lead to stock outs and/or wastage; removing them can improve product profitability by 10 to 30\% \citep{metters1997quantifying}.

We know the bullwhip effect is caused by the feedback policies (ordering policies) that firms use to manage their stocks \citep{lee1997distortion,sterman1989beergame} interacting with each other and with the transport and production delays \citep{dejonckheere2003control,chen2000quantifying} inherent to the system. These oscillations are present even in simple serial chains and the phenomenon becomes even more complex in more realistic supply networks where firms have multiple suppliers and customers coupled through a transport network with heterogeneous delays \citep{ouyang2010bullwhip,sipahi2009stability}.

Historically, the bulk of works on mitigating the bullwhip effect (reviewed
in \cite{sarimveis2008dynamic}) focused on modifying the ordering policy primarily in
serial chains using e.g., ideas from classical control
\cite{dejonckheere2003control,disney2002discrete,gaalman2006fullstate,hosoda2006altruistic},
predictive control \cite{fu2020cooperative}, and robust control
\cite{li2024mitigating,qiu2015stabilization}. There are works that analyze
the bullwhip effect in more complex supply networks through simulation
studies \cite{dominguez2015structure,ignaciuk2020quantifying} or by bounding its worst-case gain
\cite{ouyang2006characterization,ouyang2010bullwhip}, and several study how
to mitigate it by adjusting the ordering policies
\cite{wei2015exploring,zemzam2019stabilization} . None of these works address
network or adjusting delays; the network is taken as given.

Beyond the bullwhip effect, there are works that design the supply network
to optimize ``static'' performance metrics, e.g., operating costs
\cite{melo2009facility}, expected costs after a facility failure
\cite{snyder2005reliability}, and resilience to loss of nodes or edges
\cite{perera2017network,ghosh2006growing}. In the network-systems
literature the network is designed for dynamic performance: edge weights
and topology are chosen to minimize the $\mathcal{H}_2$ or
$\mathcal{H}_\infty$ norm of a consensus network
\cite{summers2015topology,ghaedsharaf2019performance,farhat2021hinf}. Those results assume symmetric consensus dynamics without per-node
feedback or heterogeneous delays which do not hold in the context of supply networks.

In this paper, we study how to mitigate the bullwhip effect in supply chains with heterogeneous delays and linear ordering policies with a network structure given by a directed acyclic graph. We model the network as a structured interconnection of inventory-ordering loops and quantify the bullwhip effect using the $\mathcal{H}_\infty$ norm of the network transfer function, building on results from consensus processes in network systems \citep{ghaedsharaf2019performance,farhat2021hinf,welikala2025consensus}. Our contributions are threefold:
\begin{enumerate}
  \item We show that the transfer function between demand and orders/inventory at any two nodes can be decomposed into a sum of serial supply chains along all paths connecting the nodes
  \item We illustrate how the bullwhip effect can be reduced by rerouting goods through the network and provide an frequency domain analysis of the effect
  \item We provide an algorithm for computing approximate gradients of the bullwhip metric and illustrate how it can be used to optimize complex networks by rerouting goods and adjusting delays using a 14-node example network.
\end{enumerate}

 The nearest works in the literature decompose a supply network into its
eigenmodes, so that amplification is resonance of a mode near its
natural frequency \cite{helbing2004physics}, with a single transport lag
shared by every route setting each mode's stability
\cite{sipahi2009stability}; delay-dependent margins have also been computed
for two-echelon chains \cite{hernandezsantos2026stability}. Those results
assume a single delay common to every route, which real networks do not
have, and attribute amplification to the network as a whole.
Darmawan et al.~\citep{darmawan2025scnd} design the network, including lead times, with
bullwhip costs in the objective, but compute the bullwhip effect for a
specific demand model and assign each node a single lead time even when it has several suppliers. Decomposing by path admits heterogeneous inbound delays at each node, where
mixing routes with different delays can cancel amplification, bounds it for
any demand through the $\mathcal{H}_\infty$ norm, and locates it in specific
routes of the topology.

The rest of the paper is organized as follows: Section~\ref{sec:setting} introduces the network model, Section~\ref{sec:analysis}
analyzes how the network structure sets amplification and stability
margins, Section~\ref{sec:design} optimizes that structure, and
Section~\ref{sec:conclusions} concludes and discusses future extensions.

\section{Problem Setting}
\label{sec:setting}

We model a supply network for a fungible good as a directed acyclic graph (DAG)
$\mathcal{G} = (\Nc, \Ec, \tau)$. The node set
$\Nc = \{v_i\}_{i=1}^{n}$ is the set of firms that consume, distribute, and
supply the good, and the edge set
\begin{equation}
  \Ec \subset \bigl\{ (u,v) : u, v \in \Nc,\ u \neq v \bigr\}
  \label{eq:edges}
\end{equation}
is the set of supply routes: edge $(i,j) \in \Ec$ runs from customer $i$
to supplier $j$, so $\Sc(i) = \{j : (i,j) \in \Ec\}$ and
$\Cc(i) = \{k : (k,i) \in \Ec\}$ are the supplier (successor) and customer (precursor) sets of node
$i$. The map $\tau : \Ec \to \mathbb{R}_{> 0}$ assigns each edge a transport delay $\tau_{ij} > 0$.
\emph{Sources} $\Sc_0 = \{j \in \Nc : \Sc(j) = \emptyset\}$ have no suppliers and place no orders; exogenous demand $v_i(t) \ge 0$ may enter at any ordering node. Every node but the sources places orders;
these are the $N = |\Nc \setminus \Sc_0|$ ordering nodes.

\begin{figure}
\centering
\resizebox{0.78\linewidth}{!}{%
\begin{tikzpicture}[>=latex, font=\small,
  nd/.style={circle, draw, minimum size=20pt, inner sep=1pt},
  goods/.style={->, thick},
  orders/.style={->, red, dashed}]
  \node[nd, fill=black!8] (D)  at (0,0)      {$D$};
  \node[nd, fill=black!8] (M1) at (2.9,0)    {$M_1$};
  \node[nd, fill=black!8] (M2) at (5.8,1.2)  {$M_2$};
  \node[nd, fill=black!8] (M3) at (5.8,-1.2) {$M_3$};
  \node[nd] (S) at (8.7,0) {$S$};
  \draw[->] (-1.5,0) -- node[above]{$v$} (D);
  \draw[goods] (M1) to[bend left=15] node[below=1pt]{$1,\ \tau=2$} (D);
  \draw[goods] (M2) to[bend right=15] node[above=1pt, sloped]{$w,\ \tau=2$} (M1);
  \draw[goods] (M3) to[bend left=15] node[below=1pt, sloped]{$1-w,\ \tau=7$} (M1);
  \draw[goods] (S)  to[bend right=15] node[above=1pt, sloped]{$1,\ \tau=2$} (M2);
  \draw[goods] (S)  to[bend left=15] node[below=1pt, sloped]{$1,\ \tau=3$} (M3);
  \draw[orders] (D)  to[bend left=15] (M1);
  \draw[orders] (M1) to[bend right=15] (M2);
  \draw[orders] (M1) to[bend left=15] (M3);
  \draw[orders] (M2) to[bend right=15] (S);
  \draw[orders] (M3) to[bend left=15] (S);
  \begin{scope}[shift={(0.2,-2.35)}]
    \draw[goods] (0,0) -- (0.8,0) node[right, black]{goods};
    \draw[orders] (3.0,0) -- (3.8,0) node[right, black]{orders};
  \end{scope}
\end{tikzpicture}}
\caption{Five-node supply network with per-edge flow weights and delays.
Solid arrows are goods, flowing from supplier to customer; dashed red
arrows are orders, flowing from customer to supplier. Weights and delays
label the route they belong to. Filled nodes place orders; exogenous demand enters at $D$; $S$ is a source.}
\label{fig:running}
\end{figure}
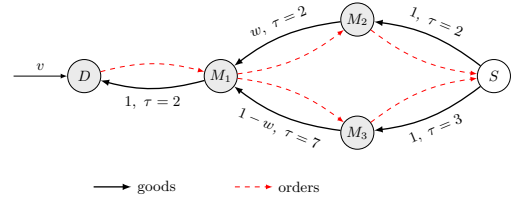

A node's demand is its customers' weighted orders plus exogenous
demand,
\begin{equation}
  d_j = \sum_{k \in \Cc(j)} w_{kj}\, o_k \;+\; v_j,
  \label{eq:coupling}
\end{equation}
with $v_j$ the exogenous part.  To satisfy its demand, a node must decide how to split its orders $o_i(t)$ among its suppliers. It does so using the weights 
$w : \Ec \to [0,1]$ which represent the fraction of node $i$'s orders routed along $j \in \Sc(i)$. The choice of these weights control how goods flow through the network and must satisfy
\begin{equation}
  \sum_{j \in \Sc(i)} w_{ij} = 1, \qquad w_{ij} \ge 0.
  \label{eq:weights}
\end{equation}

Each ordering node $i$ has an inventory position $\IP_i(t)$ (stock on hand, net of backlog) and a supply line $\SL_i(t)$ (stock ordered, not yet arrived). These goods arrive after their respective delays $\tau_{ij}$ and sum together to form the total incoming delivery (Figure~\ref{fig:loop_physical}). The combined dynamics for node $i$ are
\begin{subequations} \label{eq:multi}
\begin{align}
  \dot{\IP}_i(t) &=
    \textstyle\sum_{j \in \Sc(i)} w_{ij}\, o_i(t-\tau_{ij}) - d_i(t),
    \label{eq:inv_multi}\\
  \SL_i(t) &= \textstyle\sum_{j \in \Sc(i)} w_{ij}
    \int_{t-\tau_{ij}}^{t} o_i(m)\,\mathrm{d}m, \label{eq:sl_multi}\\
  o_i(t) &= d_i(t) - K_i\bigl(\IP_i(t) + \beta_i\,\SL_i(t)\bigr),
    \label{eq:policy_multi}\\
  d_i &= \sum_{k \in \Cc(i)} w_{ki}\, o_k \;+\; v_i.
\end{align}
\end{subequations}
The ordering policy \eqref{eq:policy_multi} is the commonly used order-up-to (OUT) rule \cite{dejonckheere2003control,disney2002discrete,hosoda2006altruistic} which applies proportional feedback to the stock error with local gain $K_i > 0$. The credit $\beta_i \in [0,1)$ determines how much of the supply line the node considers when ordering. At $\beta=0$ in-transit stock is ignored and the node will reorder stock that is already on its way. At $\beta=1$ the stock that is in-transit will be considered fully and no amplification will occur at any delay. This behaviour is not observed in practice however, behavioural studies suggest a typical value of $\beta\approx 0.34$ \cite{sterman1989beergame}, and any $\beta<1$ will amplify.

\begin{rem}[Steady state and setpoints]
\label{rem:ss}
Under constant demand $d_i(t) \equiv \bar d_i$, the equilibrium of \eqref{eq:multi} is
\begin{equation*}
o_i = \bar d_i, \quad
\SL_i = \bar d_i \textstyle\sum_{j \in \Sc(i)} w_{ij}\tau_{ij}, \quad
\IP_i = -\beta_i\,\SL_i.
\label{eq:ss}
\end{equation*}
The equilibrium can be negative because $\IP_i$ and $\SL_i$ are deviations
from setpoints $\IP^*_i$ and $\SL^*_i$. The model is linear in these
deviations, so the results below hold for demand perturbations small enough
that orders stay non-negative and no supplier or route runs out of stock or
capacity.
\end{rem}
\begin{figure}
\centering
\resizebox{0.9\linewidth}{!}{%
\begin{tikzpicture}[>=latex, font=\small,
  blk/.style={draw, minimum height=16pt, minimum width=40pt, inner sep=3pt},
  sm/.style={draw, circle, minimum size=11pt, inner sep=0pt}]
\node (din) at (-2.2,0.9) {$d_i$};
\node[sm] (so) at (-1.0,0.9) {};
\coordinate (br) at (-0.1,0.9);
\node[blk] (del1) at (1.6,1.5) {$w_{ij_1}e^{-s\tau_{ij_1}}$};
\node at (1.6,0.85) {$\vdots$};
\node[blk] (del2) at (1.6,0.2) {$w_{ij_m}e^{-s\tau_{ij_m}}$};
\node[sm] (ssum) at (3.3,0.9) {$\Sigma$};
\node[sm] (sip) at (4.7,0.9) {};
\node[blk] (int) at (6.2,0.9) {$\dfrac{1}{s}$};
\node[blk] (slb) at (2.0,-1.2) {$\sum_j w_{ij}\dfrac{1-e^{-s\tau_{ij}}}{s}$};
\node[blk] (bet) at (4.8,-1.2) {$\beta_i$};
\node[sm] (sy) at (7.6,-1.2) {};
\node[blk] (K) at (2.0,-2.6) {$K_i$};
\draw[->] (din) -- node[pos=0.78,above]{\scriptsize$+$} (so);
\draw[->] (so) -- node[above, pos=0.4]{$o_i$} (br) -- (0.4,0.9) |- (del1);
\draw[->] (0.4,0.9) |- (del2);
\draw[->] (del1) -| (ssum);
\draw[->] (del2) -| (ssum);
\draw[->] (ssum) -- node[above=2pt, pos=0.4]{\scriptsize arrivals} node[pos=0.9,above=1pt]{\scriptsize$+$} (sip);
\draw[->] (sip) -- (int);
\draw[->] (int.east) -- ++(0.9,0) node[above left]{$\IP_i$} coordinate (ipout);
\draw[->] (ipout) |- node[pos=0.96,above=1pt]{\scriptsize$+$} (sy);
\draw[->] (-2.2,2.3) node[left]{$d_i$} -| node[pos=0.97,left=1pt]{\scriptsize$-$} (sip);
\draw[->] (br) |- (slb);
\draw[->] (slb) -- node[above]{$\SL_i$} (bet);
\draw[->] (bet) -- node[pos=0.92,above=1pt]{\scriptsize$+$} (sy);
\draw[->] (sy) |- node[pos=0.75,above]{$y_i = \IP_i + \beta_i\,\SL_i$} (K);
\draw[->] (K) -| node[pos=0.97,right]{\scriptsize$-$} (so);
\end{tikzpicture}}
\caption{The ordering loop at node $i$, drawn from \eqref{eq:multi}:
orders split by the weights, transit each route's delay, and the
arrivals sum into the inventory balance; the supply line accumulates
orders placed but not yet arrived; the policy feeds back the credited
stock $y_i = \IP_i + \beta_i\,\SL_i$.}
\label{fig:loop_physical}
\end{figure}
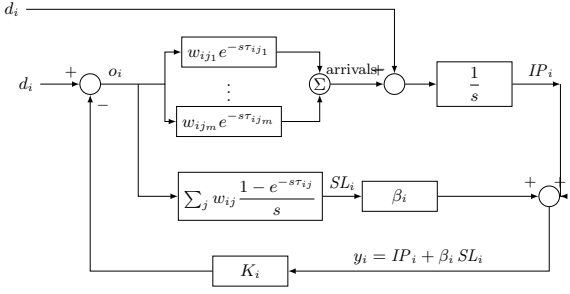

Throughout this paper we assume the following.
\begin{itemize}
\item[\textbf{A1}] $\mathcal{G}$ is a directed acyclic graph (DAG) with $\tau_{ij} > 0$;
\item[\textbf{A2}] every ordering node follows the policy
\eqref{eq:multi} with $K_i > 0$,
 $\beta_i \in [0,1)$, and weights $w_{ij}$ satisfying
\eqref{eq:weights};
\item[\textbf{A3}] the nodes are indexed in a topological order, so that
$(c,j) \in \Ec$ implies $c < j$.
\end{itemize}

\begin{rem}[Discrete time]
\label{rem:dt} Discrete-time analogs of the main results follow readily. Sampling with integer delays replaces the integrator $1/s$ by $1/(z-1)$ and the delay $e^{-s\tau}$ by $z^{-\tau}$ and every result that follows carries over with the Laplace transform replaced by the $z$-transform. The largest difference is that optimizing over integer delays $z^{-\tau}$ becomes a mixed-integer problem.
\end{rem}

\section{Network analysis}
\label{sec:analysis}

We first reduce a single node to a feedback loop, then study how the nodes
interconnect.

\subsection{Nodal dynamics}
\label{sec:node_dyn}

The following shows that the nodal dynamics \eqref{eq:multi} reduce to a single feedback loop (Figure~\ref{fig:loop_tf}), from which we derive the transfer function of each node. 

\begin{thm}[Nodal responses]
\label{thm:node}
Suppose A2 holds. Then the dynamics \eqref{eq:multi} of node $i$ are
the feedback loop of Figure~\ref{fig:loop_tf} with
\begin{subequations}
\begin{equation}
\begin{gathered}
  W_i(s) = \beta_i + (1-\beta_i) \sum_{j \in \Sc(i)}
  w_{ij}\,e^{-s\tau_{ij}},\\
  L_i(s) = \frac{K_i\,W_i(s)}{s},
\end{gathered}
  \label{eq:W_i}
\end{equation}
and the demand-to-order and demand-to-inventory responses of the node
are
\begin{align}
  T_i(s) &:= \frac{o_i(s)}{d_i(s)}
    = \frac{s+K_i}{s+K_i\,W_i(s)},
  \label{eq:Ti}\\
  G_i(s) &:= \frac{\IP_i(s)}{d_i(s)}
    = \frac{\bigl(\sum_{j \in \Sc(i)} w_{ij}\,e^{-s\tau_{ij}}\bigr)\,
      T_i(s) - 1}{s},
  \label{eq:Gi}
\end{align}
\end{subequations}
where $d_i$ is the aggregate demand \eqref{eq:coupling} of node $i$.
\end{thm}

\begin{pf}
Taking the Laplace transform of
\eqref{eq:inv_multi}--\eqref{eq:sl_multi} at zero initial conditions,
\begin{align}
  s\,\IP_i &= \Bigl(\textstyle\sum_{j} w_{ij}\,e^{-s\tau_{ij}}\Bigr) o_i - d_i,
  \label{eq:ip_z}\\
  \SL_i &= \textstyle\sum_{j} w_{ij}\,\frac{1 - e^{-s\tau_{ij}}}{s}\, o_i.
  \label{eq:sl_z}
\end{align}

Substituting these into the credited stock
$y_i = \IP_i + \beta_i\,\SL_i$ of \eqref{eq:policy_multi} yields
\begin{equation}
  s\,y_i = W_i(s)\, o_i - d_i,
  \qquad
  o_i = d_i - K_i\, y_i,
  \label{eq:y_solved}
\end{equation}
which is the loop of Figure~\ref{fig:loop_tf}. Eliminating $y_i$,
\begin{equation*}
  \bigl(s+K_i\,W_i(s)\bigr)\, o_i = \bigl(s+K_i\bigr)\, d_i,
\end{equation*}
yields \eqref{eq:Ti}, and dividing \eqref{eq:ip_z} by $d_i$ and
substituting $o_i = T_i d_i$ yields \eqref{eq:Gi}. \qed
\end{pf}

The two responses share the denominator $s + K_i W_i(s)$, and therefore
the same poles: the $1/s$ in \eqref{eq:Gi} is cancelled by a numerator that
vanishes at $s = 0$, where $W_i(0) = T_i(0) = 1$. The edge weights and delays
enter only through $W_i(s)$.
\begin{cor}[Single route]
\label{cor:single}
For a node with a single supply route of delay $\tau$, we have that $W(s) = \beta + (1-\beta)\,e^{-s\tau}$, $T(s) = \frac{s+K}{s+K W(s)}$, and $G(s) = \frac{e^{-s\tau}\,T(s) - 1}{s}$. 
\end{cor}

\noindent
At $\beta = 0$ the loop transfer function is $L_i(s) = K e^{-s\tau}/s$: the
model dynamics collapse into a proportional control loop for an integrator
with a delay.

\begin{figure}
\centering
\resizebox{0.72\linewidth}{!}{%
\begin{tikzpicture}[>=latex, font=\small,
  blk/.style={draw, minimum height=16pt, minimum width=36pt, inner sep=3pt},
  sm/.style={draw, circle, minimum size=11pt, inner sep=0pt}]
\node (din) at (-2.4,0) {$d_i$};
\node[sm] (so) at (-1.2,0) {};
\node[blk] (W) at (0.5,0) {$W_i(s)$};
\node[sm] (sy) at (2.1,0) {};
\node[blk] (int) at (3.7,0) {$\dfrac{1}{s}$};
\node[blk] (K) at (0.5,-1.5) {$K_i$};
\draw[->] (din) -- node[pos=0.35,above]{\scriptsize$+$} (so);
\draw[->] (so) -- node[above]{$o_i$} (W);
\draw[->] (W) -- node[pos=0.75,above]{\scriptsize$+$} (sy);
\draw[->] (sy) -- (int);
\draw[->] (int.east) -- ++(0.9,0) node[above left]{$y_i$} coordinate (yout);
\draw[->] (yout) |- (K);
\draw[->] (K) -| node[pos=0.97,right]{\scriptsize$-$} (so);
\draw[->] (-2.4,1.2) node[left]{$d_i$} -| node[pos=0.93,left]{\scriptsize$-$} (sy);
\end{tikzpicture}}
\caption{The nodal dynamics simplify into a SISO negative feedback loop
with loop transfer function $L_i(s) = K_i\,W_i(s)/s$ \eqref{eq:W_i}.}
\label{fig:loop_tf}
\end{figure}
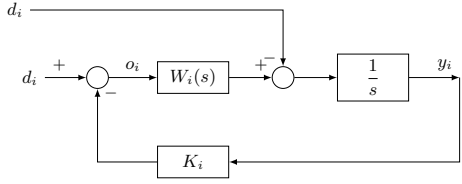

\subsection{Network dynamics}
\label{sec:composition}

To determine the bullwhip effect across the network, we must analyze how single-node dynamics couple through the network structure. The nodes connect through the demand aggregation \eqref{eq:coupling}, where one node's orders become another's demand. Defining the transfer functions
\begin{equation}
T_{\ell \to j}(s) \;:=\; \frac{o_j(s)}{v_\ell(s)},
\text{ such that }
o_j = \sum_{\ell \in \Nc \setminus \Sc_0} T_{\ell \to j} v_\ell
\label{eq:Tji_def}
\end{equation}
and stacking the orders $o = (o_j)_{j \in \Nc \setminus \Sc_0}$ and exogenous demands $v = (v_\ell)_{\ell \in \Nc \setminus \Sc0}$ yields the square network transfer matrix,
\begin{equation}
o = T(s)\, v,
\qquad
T_{j\ell}(s) = T_{\ell \to j}(s).
\label{eq:matrixT}
\end{equation}
We quantify the bullwhip effect using the $\mathcal{H}_\infty$ norm \cite{ouyang2010bullwhip} of network transfer matrix \eqref{eq:matrixT}
\begin{equation}
\|T\|_\infty = \sup_{\omega}\ \bar\sigma\bigl(T(\im\omega)\bigr),
\label{eq:hinf_def}
\end{equation}
where $\bar\sigma(\cdot)$ is the largest singular value. This is our primary bullwhip metric. The classical bullwhip measure is the variance ratio
$\text{Var}(o)/\text{Var}(d)$, which under white-noise demand becomes the
squared $\mathcal{H}_2$ norm $\|T\|_2^2$
\citep{dejonckheere2003control,zhou1996robust}. The policy
\eqref{eq:policy_multi} passes demand straight to orders, so $T$ is not
strictly proper and $\|T\|_2$ is infinite \citep{zhou1996robust}; we use
$\|T\|_\infty$ instead \citep{li2024mitigating}. Transient responsiveness is measured by the gain crossover frequency
\begin{equation}
\omega_g(w, K)
= \min\bigl\{\omega > 0 \;:\; K\,|W(\im\omega)| = \omega\bigr\},
\label{eq:wg_def}
\end{equation}
of the loop, where $|L_i(\im\omega)| = 1$ which we use in place of bandwidth since the feedthrough term in
\eqref{eq:policy_multi} gives $|T(\im\infty)| = 1$, and $|T(0)| = 1$, so $|T|$
never falls $3$~dB below its low-frequency value limit.  For a network we report $\omega_g=\min_i \omega_{g,i}$ the crossover of the slowest node. Local robustness is quantified by the disk margin of node
$i$'s loop \citep{seiler2020introduction},
\begin{equation}
\alpha_i(w, K)
= \Bigl(\sup_{\omega \ge 0}
\bigl|S_i(\im\omega) - \tfrac{1}{2}\bigr|\Bigr)^{-1},
\qquad S_i = \frac{1}{1+L_i},
\label{eq:disk}
\end{equation}
with $L_i$ as in \eqref{eq:W_i}.

To evaluate the network-wide response, we decompose the topology into directed paths of individual ordering loops. Exogenous demand from node $\ell$ reaches node $j$ along the set of paths
\begin{gather}
  \squeeze{\Ppaths{\ell}{j}
  = \bigl\{ (\ell, \dots, l, {l+1}, \dots, j) ~|~
    (l, {l+1}) \in \Ec \bigr\},}
  \label{eq:paths_def}
\end{gather}
and we define the transfer function of a path $p \in \Ppaths{\ell}{j}$ as
\begin{equation}
  \Tp{p}(s) = \Bigl(\Pi_{e\in \pe{p}} w_e \Bigr)
  \Bigl(\Pi_{n\in (p \setminus \{n_j\})} T_n\Bigr),
  \label{eq:pathTF}
\end{equation}
where $\pe{p}$ is the set of edges of $p$, with the convention
$\Tp{p}(s) = 1$ whenever $|p| = 1$. The following theorem shows how the network
dynamics decomposes along paths.

\begin{thm}[Path composition]
\label{thm:path}
Let A1--A3 hold. Then for every pair of ordering nodes $\ell, j \in \mathcal{N} \setminus \mathcal{S}_0$,
\begin{equation}
  T_{\ell \to j}(s)
  = T_j(s) \sum_{p \in \Ppaths{\ell}{j}} \Tp{p}(s),
  \label{eq:path_sum}
\end{equation}
with the convention that $T_{\ell \to j}(s) = 0$ if
$\Ppaths{\ell}{j} = \emptyset$.
\end{thm}

\begin{pf}
See Appendix~\ref{app:path}.
\end{pf}
Substituting \eqref{eq:path_sum} into \eqref{eq:matrixT}  gives every entry of the network transfer matrix, for row $j$ and column $\ell$,
\begin{equation}
  T(s) = \Bigl[\, T_j(s) \sum_{p \in \Ppaths{\ell}{j}} \Tp{p}(s) \Bigr]_{j,\ell \in \Nc \setminus \Sc_0}.
  \label{eq:T_paths}
\end{equation}

Theorem~\ref{thm:path} decouples the internal dynamics of the individual nodes
from the larger network structure. Exogenous demand propagates through the
network by multiplying sequential nodes along a path and summing across
parallel paths, so the network-wide bullwhip effect is the accumulation of
local ordering behaviours, and the poles of the network are exactly the poles
of the nodes it contains. Figure~\ref{fig:paths} illustrates the application of
Theorem~\ref{thm:path} to a seven-node network.

\begin{figure}[!htb]
\centering
\resizebox{0.88\linewidth}{!}{%
\begin{tikzpicture}[>=latex, font=\small,
  nd/.style={circle, draw, fill=black!8, minimum size=20pt, inner sep=1pt},
  wl/.style={fill=white, inner sep=1pt, font=\scriptsize}]
  \node[nd] (n1) at (0,0)      {$1$};
  \node[nd] (n2) at (2.4,1.5)  {$2$};
  \node[nd] (n3) at (2.4,0)    {$3$};
  \node[nd] (n4) at (2.4,-1.7) {$4$};
  \node[nd] (n5) at (5.6,1.1)  {$5$};
  \node[nd] (n6) at (5.6,-0.9) {$6$};
  \node[nd] (n7) at (8.2,0.1)  {$7$};
  \draw[->, thin] (n1) -- node[wl, above, sloped, pos=0.5]{$w_{12}$} (n2);
  \draw[->, thin] (n1) -- node[wl, above, pos=0.5]{$w_{13}$} (n3);
  \draw[->, thin] (n1) -- node[wl, below, sloped, pos=0.5]{$w_{14}$} (n4);
  \draw[->, thin] (n2) -- node[wl, above, sloped, pos=0.45]{$w_{25}$} (n5);
  \draw[->, thin] (n2) -- node[wl, above, sloped, pos=0.72]{$w_{26}$} (n6);
  \draw[->, thin] (n3) -- node[wl, above, sloped, pos=0.30]{$w_{35}$} (n5);
  \draw[->, thin] (n3) -- node[wl, below, sloped, pos=0.35]{$w_{36}$} (n6);
  \draw[->, thin] (n4) -- node[wl, below, sloped, pos=0.5]{$w_{46}$} (n6);
  \draw[->, thin] (n5) -- node[wl, above, sloped, pos=0.5]{$w_{57}$} (n7);
  \draw[->, thin] (n6) -- node[wl, below, sloped, pos=0.5]{$w_{67}$} (n7);
\end{tikzpicture}}\\[8pt]
\newcommand{\pathpanel}[2]{%
\begin{tikzpicture}[baseline=0pt]
  \begin{scope}[scale=0.30]
    \coordinate (n1) at (0,0);
    \coordinate (n2) at (2.4,1.5);
    \coordinate (n3) at (2.4,0);
    \coordinate (n4) at (2.4,-1.7);
    \coordinate (n5) at (5.6,1.1);
    \coordinate (n6) at (5.6,-0.9);
    \coordinate (n7) at (8.2,0.1);
    \foreach \a/\b in {n1/n2, n1/n3, n1/n4, n2/n5, n2/n6, n3/n5,
                       n3/n6, n4/n6, n5/n7, n6/n7}
      \draw[black!25, line width=0.5pt] (\a) -- (\b);
    \draw[#1, line width=1.5pt, line cap=round, rounded corners=2pt] #2;
    \foreach \p in {n1, n2, n3, n4, n5, n6, n7}
      \fill[black!50] (\p) circle (2.2pt);
  \end{scope}
\end{tikzpicture}}%
\setlength{\tabcolsep}{5pt}%
\resizebox{\linewidth}{!}{%
\begin{tabular}{ccccc}
\pathpanel{violet!60}{(n1) -- (n2) -- (n5) -- (n7)} &
\pathpanel{magenta!60}{(n1) -- (n3) -- (n5) -- (n7)} &
\pathpanel{teal!70}{(n1) -- (n2) -- (n6) -- (n7)} &
\pathpanel{orange!80}{(n1) -- (n3) -- (n6) -- (n7)} &
\pathpanel{blue!60}{(n1) -- (n4) -- (n6) -- (n7)} \\[2pt]
\colorbox{violet!18}{\scriptsize $(1,2,5,7)$} &
\colorbox{magenta!16}{\scriptsize $(1,3,5,7)$} &
\colorbox{teal!20}{\scriptsize $(1,2,6,7)$} &
\colorbox{orange!22}{\scriptsize $(1,3,6,7)$} &
\colorbox{blue!13}{\scriptsize $(1,4,6,7)$}
\end{tabular}}
{\scriptsize\setlength{\fboxsep}{6pt}
\begin{align*} 
  T_{1\to 7} ={}&
  \colorbox{violet!18}{\tiny$T_1 w_{12} T_2 w_{25} T_5 w_{57} T_7$}
  + \colorbox{magenta!16}{\tiny$T_1 w_{13} T_3 w_{35} T_5 w_{57} T_7$}
  \\[-3pt]
  &+ \colorbox{teal!20}{\tiny$T_1 w_{12} T_2 w_{26} T_6 w_{67} T_7$}
  + \colorbox{orange!22}{\tiny$T_1 w_{13} T_3 w_{36} T_6 w_{67} T_7$}
  \\[-3pt]
  &+ \colorbox{blue!13}{\tiny$T_1 w_{14} T_4 w_{46} T_6 w_{67} T_7$}.
\end{align*}}
\caption{The five paths from node $1$ to node $7$. Top: the network, orders
drawn customer to supplier and labelled by their weights. Bottom: one panel per path, coloured as its term in the
expansion above, which is Theorem~\ref{thm:path} applied to this network.}
\label{fig:paths}
\end{figure}
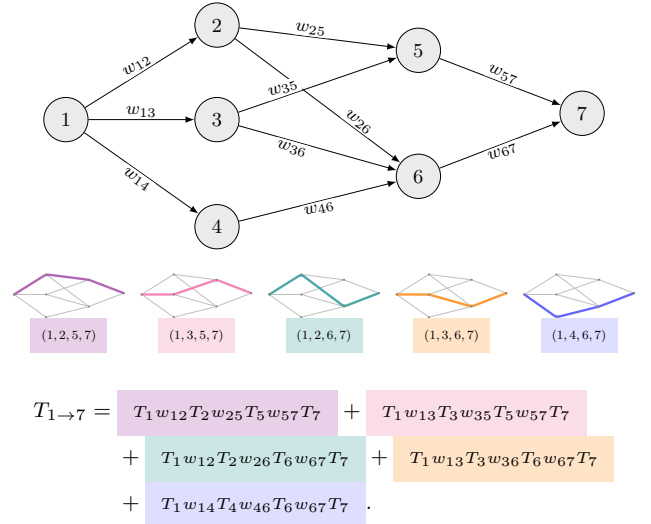
We also have some immediate corollaries:
\begin{cor}[Serial chain]
\label{cor:serial}
If $\Ppaths{\ell}{j} = \{p\}$, so that a single path runs from $\ell$ to $j$,
then $T_{\ell \to j} = \bigl(\prod_{e \in \pe{p}} w_e\bigr)\prod_{n \in p} T_n$.
\end{cor}

\begin{cor}[Causality]
\label{cor:causal}
Under A1--A3, $T$ is lower triangular and $o_j = \sum_{\ell \in \Nc \setminus \Sc_0, \ell \le j} T_{\ell \to j} v_\ell$.
\end{cor}
This corollary captures the intuition that orders should only depend on upstream demand and is a consequence of the DAG structure.

\begin{cor}[Stability]
\label{cor:stability}
Under A1--A3, the poles of $T_{\ell \to j}$ satify $\mathrm{Poles}(T_{\ell \to j}) =
  \bigcup_{p \in \Ppaths{\ell}{j}, n \in p}
  \mathrm{Poles}(T_n)
  = \bigcup_{p \in \Ppaths{\ell}{j}, n \in p}
  \bigl\{\, s ~|~ s + K_n\,W_n(s) = 0 \,\bigr\}$
so the network is internally stable if and only if every node $T_n$ is stable.
\end{cor}

These results hold for any linear ordering policy, since the proof uses only
$o_i = T_i d_i$; a nonlinear policy must be linearized first
(Remark~\ref{rem:ss}).

\subsection{Stability margins}
\label{sec:margins}

Since network stability decouples node by node
(Corollary~\ref{cor:stability}), we quantify network robustness through each
node's loop $L_i$. Each node's dynamics have the form shown in
Figure~\ref{fig:loop_tf}, which is amenable to the classical frequency-domain
margins and performance metrics. For the case of a single supplier several of
these have closed forms.

\begin{thm}[Margins, single route]
\label{thm:margins_single}
Let a node have a single supply route of delay $\tau$, credit $\beta = 0$, and gain $K > 0$. Then the gain marin, phase margin, and crossover frequency are $g_m(\tau, K) = \frac{K^*(\tau)}{K}$ with $K^*(\tau) = \frac{\pi}{2\tau}$, $\varphi_m(\tau, K) = \frac{\pi}{2} - K\,\tau$, and $\omega_g(K) = K$ respectively.
\end{thm}

\begin{pf}
With one route, $W(\im\omega) = e^{-\im\omega\tau}$, so
\begin{equation}
  |L(\im\omega)| = \frac{K}{\omega},
  \qquad
  \angle L(\im\omega) = -\omega\tau - \frac{\pi}{2},
  \label{eq:L_polar_single}
\end{equation}
both are strictly decreasing in $\omega$. The phase crosses $-\pi$ at
$\omega_c = \pi/(2\tau)$, so
$g_m = 1/|L(\im\omega_c)| = \pi/(2 K \tau)$; the gain crosses one
at $\omega_g = K$, where
$\varphi_m = \pi + \angle L(\im\omega_g) = \pi/2 - K\tau$. \qed

\end{pf}
Theorem~\ref{thm:margins_single} establishes the trade-off between transient
responsiveness (bandwidth) and the bullwhip effect (peak amplification). Since
the crossover frequency/bandwith is $K$ and the margins  depend on $K\tau$, increasing the gain for faster
responsiveness degrades the margins, and at $K\tau = \pi/2$ the phase margin
vanishes and $\|T\|_\infty$ diverges. A longer delay forces a proportional
reduction in bandwidth to maintain the same margins.

\subsection{Margins with multiple suppliers}
\label{sec:margins_multi}

With a single supplier, stability margins are dictated
by the delay of that one edge. With multiple suppliers we lose these closed
forms but gain degrees of freedom we can use to our advantage. Both the
bandwidth and the stability margin are governed by the choice of routing
weights $w$ and proportional gain $K$, and improving one inherently limits
the other. We map
this Pareto front of optimal trade-offs by solving the following problem:
\begin{equation}
\begin{aligned}
  \omega_g^*(\alpha_{\min})
  = \max_{w,\,K}\;\; & \omega_g(w, K)
  \\
  \text{s.t.}\;\;
  & \textstyle\sum_j w_j = 1, \quad w \ge 0,
  \\
  & \alpha(w, K) \ge \alpha_{\min}.
\end{aligned}
  \label{eq:pareto_prog}
\end{equation}
The resulting front is shown in Figure~\ref{fig:pareto}; the solver is
described in the supplementary material. Along the front the peak
amplification $\|T\|_\infty$ falls together with the bandwidth, so the margin
requirement fixes responsiveness and amplification at the same time.

\begin{figure}[!t]
  \centering
  \includegraphics[width=0.9\linewidth]{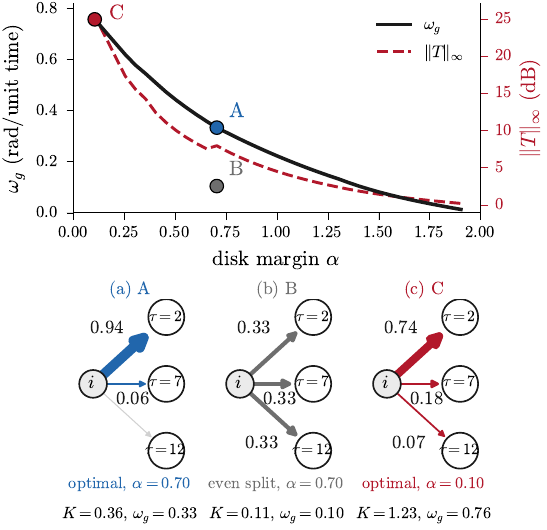}
  \caption{The responsiveness--margin front for a node with routes $\tau \in \{2, 7, 12\}$ time units at $\beta = 0$ (above) and illustrations for optimal designs at $\alpha = 0.7$ (A) and $0.1$ (C) and a suboptimal design with $\alpha =0.7$ (B) (below). The step near $\alpha=0.7$ is where the optimal split switches from two routes to three. Optimizing at $\alpha = 0.7$ leads to a $3.2 \times$ increase in bandwidth.}
  \label{fig:pareto}
\end{figure}

\begin{figure}[!t]
  \centering
  \includegraphics[width=0.95\linewidth]{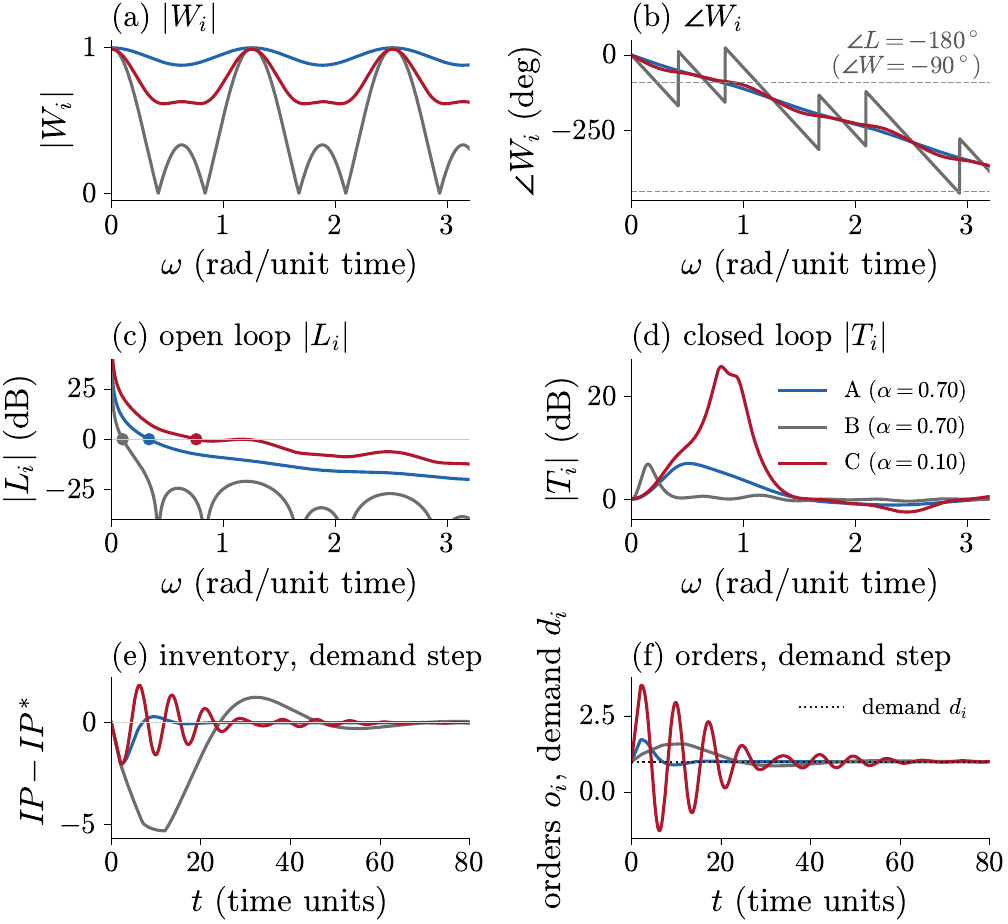}
  \caption{The frequency and time-domain responses of designs A
  , B , and C  from Figure~\ref{fig:pareto}.}
  \label{fig:pareto_bode}
\end{figure}

The resulting Pareto front shows that the optimal policy mixes routes of
different delays at every margin level. For example, designs A and B have
identical stability margins but vastly different bandwidths, as shown in their
responsiveness to a demand step (Fig.~\ref{fig:pareto_bode}f). Near $\alpha=0.7$ two routings, one using all three routes and one leaving the slowest unused, have the same bandwidth. The front is therefore continuous in $\omega_g$, but the two routings differ in peak amplification, creating the step in Figure~\ref{fig:pareto}.

However with multiple suppliers inventory arrives out of phase which results in
incoming supply signals to destructively interfere. To establish exactly why
mixing edges with different delays can be beneficial for stability, we
formalize the following interference.

\begin{prop}[Interference identity]
\label{prop:interference}
Write $W_i = \beta_i + (1-\beta_i)\,W_i^0$, where
$W_i^0(s) = \sum_{j \in \Sc(i)} w_{ij}\,e^{-s\tau_{ij}}$ is the response
from orders placed to goods arriving. Then
\begingroup\small
\begin{align}
  \bigl|W_i^0(\im\omega)\bigr|^2
  &= \textstyle\sum_j w_{ij}^2
    + 2\sum_{j<k} w_{ij} w_{ik}
      \cos\bigl(\omega(\tau_{ij}-\tau_{ik})\bigr),
  \label{eq:interference}\\
  \angle W_i^0(\im\omega)
  &= \operatorname{atan2}\Bigl(
     -\textstyle\sum_j w_{ij}\sin(\omega\tau_{ij}),\nonumber\\
  &\qquad\qquad\quad
     \textstyle\sum_j w_{ij}\cos(\omega\tau_{ij})\Bigr),
  \label{eq:Wphase}\\
  \bigl|W_i(\im\omega)\bigr|^2
  &= \beta_i^2 + 2\beta_i(1-\beta_i)\operatorname{Re} W_i^0(\im\omega)
  \nonumber\\
  &\qquad + (1-\beta_i)^2 \bigl|W_i^0(\im\omega)\bigr|^2,
  \label{eq:Wbeta}
\end{align}\endgroup
and $\bigl|W_i(\im\omega)\bigr| \le 1$ for every $\omega$ and every
$\beta_i \in [0,1)$, with equality only when all routes are in phase at
$\omega$.
\end{prop}

Distributing supply across multiple routes improves stability margins without lowering $K$ by reducing phase lag at the gain crossover and attenuating magnitude at the phase crossover. Specifically, the cross terms in \eqref{prop:interference} pull $\vert{}W_i\vert{}$ below one near the $-180^\circ$ crossings (Fig.~\ref{fig:pareto_bode}a,c), which lowers the loop gain $\vert{}L\vert{}$ and increases the gain margin. Simultaneously, the effective phase lag scales with the weighted average of the delays (Fig.~\ref{fig:pareto_bode}b), preserving the phase margin even with the inclusion of slower routes. The term $W_i$ acts as the effective delay for a node ordering from multiple suppliers. At low target disk margins (higher responsiveness), the mix of edge weights maintains closed-loop stability at proportional gains that would cause a node with only a single-route to be unstable. At high target margins (higher robustness), this attenuation sustains a higher crossover frequency than a node relying solely on the fastest supplier.

\section{Supply Chain Network Design}
\label{sec:design}

The closed-loop response of a network depends on the edge weights, edge delays, and controller gains. This section illustrates how we can use these additional degrees of freedom to optimize the network.

We look to minimize the bullwhip metric of the entire network,
\begin{equation}
  \min_{(w,\delta,K) \in \mathcal{W}_b} \bigl\|T(w, \tau^0 - \delta, K)\bigr\|_\infty.
  \label{eq:p_joint}
\end{equation}
Each edge's nominal lead time $\tau^0_e$ can be shortened by a reduction $\delta_e$, giving an operated delay $\tau_e = \tau^0_e - \delta_e$. Stacking these reductions with the gains of all nodes $K = (K_i)_i$ and routing weights $w$, the feasible set is
\begin{equation}
  \squeeze{\mathcal{W}_b = \Bigl\{ (w,\delta,K)|
  \begin{aligned}[t]
  & \textstyle\sum_{j \in \Sc(i)} w_{ij} = 1, ~
    0 \le w_{ij} \le w^{\max}_{ij},\\
  & \delta_e \ge 0, ~ \tau^0_e - \delta_e \ge \tau_{\min}, ~
    \textstyle\sum_e \delta_e \le b,\\
  & \alpha_i(w, \tau^0 - \delta, K_i) = \alpha_{\min}
    \;\; \forall\, i \Bigr\}.
  \end{aligned}}
  \label{eq:feasible}
\end{equation}

The constraints in \eqref{eq:feasible} enforce flow allocations within capacity limits $w^{\max}_{ij}$, require lead-time reductions to be nonnegative, respect a minimum delay $\tau_{\min}$, and keep total spending within the budget $b$. The disk margin $\alpha_i$ decreases in $K_i$, so the margin binds at the
optimum. The equality in \eqref{eq:feasible} selects the largest gain each node
can run, and $K$ follows from $w$ and $\tau$ rather than being chosen.

We use a 14-city European supply network as an example throughout this section, which has five tiers: two source
ports, plants, regional and local distribution, and three demand nodes, the only nodes with nonzero exogenous demand, connected by 33 edges (Fig.~\ref{fig:joint_grad}). Inventory
travels at 150~km per unit time, setting nominal lead times from $1.5$
(Zurich--Milan) to $13.5$ (Granada--Prague). Each edge also carries a unit
freight cost $c : \Ec \to \mathbb{R}_{>0}$, higher for premium fast freight
than for slow freight, so $c_e$ decreases with transit time. Only three columns of $T$ are nonzero and the norms below are taken over them.

\begin{rem}[$\mathcal{H}_2$ design]
\label{rem:h2}We focus on  $\|T\|_\infty$ for brevity but we could replace $\|T\|_\infty$ with the band-limited $\mathcal{H}_2$ norm $  \|T\|_{2,\bar\omega}^2 = \frac{1}{\pi}\int_0^{\bar\omega}
  \|T(\im\omega)\|_F^2\,\mathrm{d}\omega$
which is the classical bullwhip measure restricted to the
frequencies below $\bar\omega$ to account for the direct feedthrough terms in \eqref{eq:policy_multi} that make $\|T\|_2$ infinite.
\end{rem}

\subsection{Gradients}
\label{sec:grad}

Solving \eqref{eq:p_joint} requires the derivatives of the objective in the
design variables $\theta = (w, \delta, K)$, and the supremum in
\eqref{eq:hinf_def} is not differentiable when more than one frequency attains
the peak. Writing $\bar\sigma_k = \bar\sigma(T(\im\omega_k))$ on a grid
$\Omega = \{\omega_k\}$ truncated at $\bar\omega = \pi$, above which
$\bar\sigma$ has settled to its high-frequency value
(Fig.~\ref{fig:worstnode}a), we replace the supremum by the
Kreisselmeier--Steinhauser smooth approximation to $\max_k \bar\sigma_k$
\citep{kreisselmeier1979systematic},
\begin{equation}
  \squeeze{J_\rho(\theta) = m + \frac{1}{\rho} \log \sum_{\omega_k \in \Omega}
  \exp\bigl(\rho(\bar\sigma_k - m)\bigr),
  \quad m = \max_{k} \bar\sigma_k .}
  \label{eq:ks}
\end{equation}
Here $\rho > 0$ is the smoothing parameter. By construction $J_\rho \to \|T\|_\infty$ as $\rho \to 0$ and as the grid $\Omega$ is
refined. Differentiating \eqref{eq:ks},
\begin{equation}
  \pd{J_\rho}{\theta} = \sum_{\omega_k \in \Omega} c_k\, \pd{\bar\sigma_k}{\theta},
  \qquad
  c_k = \frac{e^{\,\rho \bar\sigma_k}}{\sum_{k'} e^{\,\rho \bar\sigma_{k'}}},
  \label{eq:chain}
\end{equation}
with weights $c_k$ concentrating on the frequencies that attain the peak. Thus computing $\nabla_\theta \|T\|_\infty$ boils down to calculating $\partial \bar\sigma(T(j\omega)) / \partial\theta$ at an arbitrary $\omega > 0$.

Under the mild assumption that $\bar\sigma_k$ is simple, i.e., strictly larger than the second singular value of $T(\im\omega_k)$, the derivative $\partial \bar\sigma_k/\partial \theta$ is well-defined and depends on the network's response to the demand pattern that $T(\im\omega_k)$ amplifies the most.

\begin{thm}[Gradient of $J_\rho$]
\label{thm:grad}
Let A1--A3 hold and suppose $\bar\sigma_k$ is simple for every $k$. Then for every
edge $(i,j) \in \Ec$ and every ordering node $i$,
\begin{subequations}\label{eq:grads}
\begin{align}
  \pd{J_\rho}{w_{ij}} &= \sum_{\omega_k \in \Omega} \operatorname{Re}
     \bigl[ (1 - \beta_i)\, e^{-\im\omega_k \tau_{ij}}\, \Lambda_i^k
     + T_j(\im\omega_k)\, \psi_j^k\, o_i^k \bigr],
  \label{eq:gw}\\
  \pd{J_\rho}{\delta_{ij}} &= \sum_{\omega_k \in \Omega} \operatorname{Re}
     \bigl[ \im\omega_k (1 - \beta_i)\, w_{ij}\,
     e^{-\im\omega_k \tau_{ij}}\, \Lambda_i^k \bigr],
  \label{eq:gd}\\
  \pd{J_\rho}{K_i} &= \sum_{\omega_k \in \Omega} \operatorname{Re}
     \Bigl[ \frac{\im\omega_k \bigl(1 - W_i(\im\omega_k)\bigr)\, T_i(\im\omega_k)\, o_i^k}
                 {(\im\omega_k + K_i)^{2}}\, \psi_i^k \Bigr],
  \label{eq:gk}
\end{align}
\end{subequations}
where $o^k = \bar\sigma_k \eta_k$ with $\eta_k$ the unit left singular vector
of $T(\im\omega_k)$ associated with $\bar\sigma_k$,
\begin{equation} \label{eq:adjoint}
  \psi_i^k = c_k\, \overline{\eta_{k,i}}
  + \sum_{j \in \Sc(i)} w_{ij}\, T_j(\im\omega_k)\, \psi_j^k,
\end{equation}
where $\overline{(\cdot)}$ the complex conjugate, $T_j \psi_j^k = 0$ for j $\in \Sc_0$, and \begin{equation} \label{eq:lambda}
  \Lambda_i^k = -\,\frac{K_i\, T_i(\im\omega_k)\, o_i^k}{\im\omega_k + K_i}\, \psi_i^k .
\end{equation}
\end{thm}

\begin{pf}
See Appendix~\ref{app:adjoint}.
\end{pf}

The gradient only depends on the demand pattern that excites the network most. When $\overline{\sigma}_k$ is simple, the peak at $\omega_k$ responds to a design change only through the network's response to that pattern, $\xi_k$, which produces the order pattern $o_k=\overline{\sigma}_k\eta_k$; $\psi^k$ is the sensitivity of the peak to those orders. 

Evaluating $\|T\|_\infty$ or the smoothed approximation \eqref{eq:ks} requires $T(\im\omega_k)$ and its singular value decomposition at every grid frequency. Computing the gradient only requires an additional backward
pass \eqref{eq:adjoint} through the network. The complexity of each pass is linear in $|\Ec|$ and independent of the number of design variables, the marginal cost of computing the $2|\Ec| + |\Nc \setminus \Sc_0|$ derivatives is linear in $|\Ec|$. The disk margin $\alpha_i$ in \eqref{eq:feasible} is differentiated analogously. Whenever $\bar\sigma_k$ is repeated the derivative is not unique, and we use the value given by any one singular vector pair, as in \cite{apkarian2006nonsmooth} which corresponds to picking a specific element of the sub-differential.

Figure~\ref{fig:joint_grad} applies Theorem~\ref{thm:grad} to compute the weight gradients to the 14-city network at an even split of edge weights across every node's suppliers.

\begin{figure}[!t]
  \centering
  \includegraphics[width=0.97\linewidth]{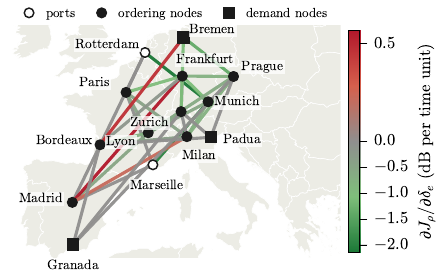}
  \caption{The 14-city network and the delay gradient \eqref{eq:gd} on it,
  evaluated at an even split of edge weights across every node's suppliers. Shortening Munich--Rotterdam lowers the peak
  the most, at $-2.1$~dB per time unit, while shortening Frankfurt--Madrid,
  one of the longest routes, raises it by $0.6$.}
  \label{fig:joint_grad}
\end{figure}

\subsection{Baselines and solution}

With expressions for the gradients in hand, we now move on to solving \eqref{eq:p_joint} for the 14-city network. We compare against two baselines that represent current practice. Neither baseline uses Theorem~\ref{thm:grad}, and both lie in
$\mathcal{W}_b$, so the three designs meet the same constraints and hold the
same disk margin.

The cost-minimal design is standard practice in supply chain network design
\citep{melo2009facility,snyder2005reliability}. It puts as much flow as
capacity allows on the cheapest edges,
\begin{equation}
  \squeeze{w^{\mathrm{cm}} \in \operatorname*{arg\,min}_{w}
  \Bigl\{ \textstyle\sum_e c_e\, w_e \;:\; \eqref{eq:weights},\
  w \le w^{\max} \Bigr\},
  \quad \delta^{\mathrm{cm}} = 0,}
  \label{eq:baseline_cm}
\end{equation}
and spends no budget on shortening.

The naive design is what a practitioner who knows transport delays impact the network dynamics, but has no
gradient, would build. It puts as much weight as possible on the fastest edges and spends the whole
budget on shortening the longest delays,
\begin{equation}
\begin{gathered}
  \squeeze{w^{\mathrm{n}} \in \operatorname*{arg\,min}_{w}
  \Bigl\{ \textstyle\sum_e \tau^0_e\, w_e \;:\; \eqref{eq:weights},\
  w \le w^{\max} \Bigr\},}\\
  \delta^{\mathrm{n}} \in \operatorname*{arg\,min}_{\delta \in \Delta_b}\
  \max_{e} \bigl(\tau^0_e - \delta_e\bigr),\\
  \Delta_b = \bigl\{ \delta \ge 0 :\ \tau^0 - \delta \ge \tau_{\min},\
  \textstyle\sum_e \delta_e \le b \bigr\} .
\end{gathered}
  \label{eq:baseline_naive}
\end{equation}

In all cases the capacity caps are $w^{\max} = 0.8$, the delay floor is $\tau_{\min} = 0.5$,
the disk margin is $\alpha_{\min} = 0.4$, and the capital budget is $b = 6$. We
solve \eqref{eq:p_joint} by sequential quadratic programming with the gradients
of Theorem~\ref{thm:grad}, from random initial splits of every node's orders,
over a uniform grid of $160$ frequencies on $[0, \pi]$. The peak is nonsmooth
in the design variables, so we optimize the smooth upper bound
$J_\rho \ge \max_k \bar\sigma_k$ of \eqref{eq:ks} and use a multistart to
handle local minima. Every reported peak is then re-evaluated on a grid of
$60\,000$ frequencies: no design moves by more than $0.07$~dB and their
ranking is unchanged (supplementary material), so neither the grid nor the
smoothing affects the comparison. Each solve of \eqref{eq:p_joint} from one start takes under $15$~s on a
single core of an Apple M3 and the full budget sweep with four starts per budget under ten
minutes.

\subsection{Results}

The results are illustrated in Figs.~\ref{fig:budget}--\ref{fig:worstnode}.

\begin{figure}[!t]
  \centering
  \includegraphics[width=0.8\linewidth]{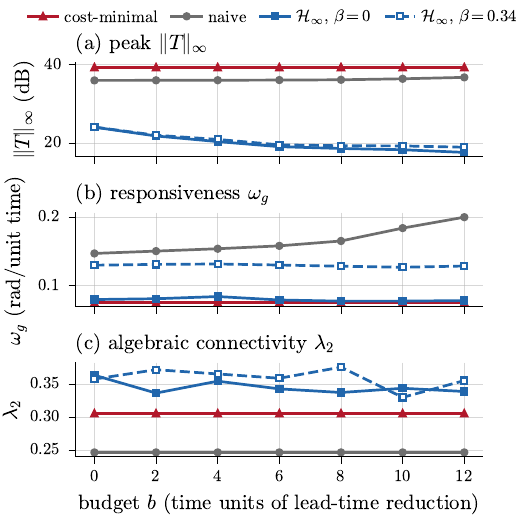}
  \caption{Performance metrics vs. lead-time budget $b$ at
  $\alpha_{\min} = 0.4$. Changes to routing alone ($b = 0$) closes $15$ of the $20$~dB between cost-minimal and
$\mathcal{H}_\infty$; the naive design converts its budget into
responsiveness rather than attenuation; and $\lambda_2$ of the
$\mathcal{H}_\infty$ design exceeds both baselines at every budget.}
  \label{fig:budget}
\end{figure}
The $\mathcal{H}_\infty$ design amplifies far less than either baseline at
every budget (Fig.~\ref{fig:budget}a). Routing changes, which are cheap
operational changes, account for most of the reduction: the
$\mathcal{H}_\infty$ design with no budget already removes $15.3$ of the
$20.3$~dB gained at $b = 6$. The $\mathcal{H}_\infty$ network leaves most
routes below their capacity cap, uses more of them per node than either
baseline, and mixes lead times at each node (Fig.~\ref{fig:joint_designs}).
It does this without giving up responsiveness: its $\omega_g$ stays at the
cost-minimal value at every budget (Fig.~\ref{fig:budget}b). Crediting the
supply line changes little: at $\beta = 0.34$ every node tolerates roughly
twice the gain at $\alpha_{\min}$, the $\mathcal{H}_\infty$ peak moves by at
most $1.3$~dB over the sweep and $0.5$~dB at $b = 6$, and the design still
improves on the cost-minimal one by $22.4$~dB.

\begin{figure}[!t]
  \centering
  \includegraphics[width=0.95\linewidth]{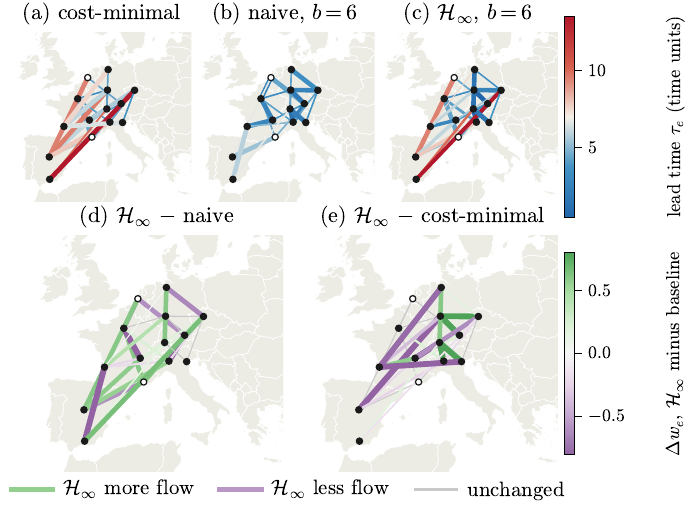}
  \caption{(a)--(c) the optimized networks, line width indicates flow weight $w_e$ and colour indicates lead time $\tau_e$; the naive and $\mathcal{H}_\infty$
  designs have spent the budget $b=6$. (d)--(e) the difference between the edge weights of the $\mathcal{H}_\infty$ design against each baseline. Both baselines put 12 of the 33 edges at the capacity cap and use two routes per node; the $\mathcal{H}_\infty$ design puts 7 edges at the cap and uses $2.2$, with $46\%$ of its flow on routes longer than the median against $80\%$ for cost-minimal and $28\%$ for
  naive.}
  \label{fig:joint_designs}
\end{figure}
Spending on lead times without the gradient buys responsiveness but no
improvement in amplification. The naive design spends its whole budget on shortening long
edges, which raises its $\omega_g$ from $0.15$ to $0.20$ as budget increases while its
$\|T\|_\infty$ does not fall (Fig.~\ref{fig:budget}a,b), because shortening
the longest routes results in near uniform delays across the network, and by
\eqref{eq:interference} the attenuation of $|W_i|$ comes from the differences
$\tau_{ij} - \tau_{ik}$. Removing them drives $|W_i|$ toward one, which raises
the loop gain.

Fig.~\ref{fig:worstnode} shows the same comparison in time. After a unit step
in demand, the inventory of the most amplified supplier oscillates for the
length of the horizon under the cost-minimal design and swings several fold
under the naive design, while the $\mathcal{H}_\infty$ design holds its
inventory near setpoint, though it settles more slowly than the naive
design; its flat response in Fig.~\ref{fig:worstnode}a is the peak being
levelled across frequencies, not over-damping.

The algebraic connectivity of the routing is the second-smallest eigenvalue
\begin{equation}
  \lambda_2(w) = \lambda_2\bigl(\mathrm{diag}(A\mathbf{1}) - A\bigr),
  \qquad A_{ij} = A_{ji} = w_{ij},
  \label{eq:lam2}
\end{equation}
of the Laplacian of the undirected routing graph $(\Nc,\Ec,w)$. It measures how well the
network stays connected when an edge is removed: it is zero when the graph
disconnects, and it grows as flow spreads over more
routes \citep{ghosh2006growing, summers2015topology}. It is larger for the
$\mathcal{H}_\infty$ design than for either baseline at every budget
(Fig.~\ref{fig:budget}c). The induced norm is largest when flow concentrates on a few routes, so minimizing it spreads flow over more of them. Designing against amplification
therefore improves resilience to lost routes as well, with no connectivity
constraint in \eqref{eq:p_joint}.

\begin{figure}[!b]
  \centering
  \includegraphics[width=0.86\linewidth]{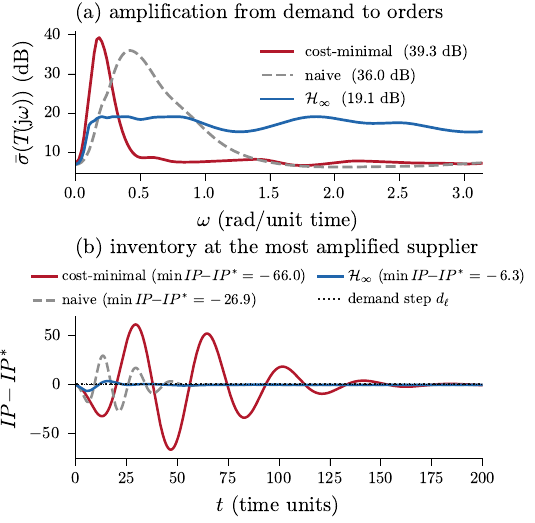}
  \caption{(a) $\bar\sigma(T(\im\omega))$ for the three designs. (b) Inventory
position at the most amplified supplier, after a unit step in
demand (dotted) at the worst case demand node: Bremen to Madrid for cost-minimal, Bremen to Munich for naive, Grenada to Munich for $\mathcal{H}_\infty$.}
  \label{fig:worstnode}
\end{figure}

\section{Conclusions}
\label{sec:conclusions}

This papers models a supply network using a DAG and proves that the transfer functions between demand and orders can be decomposed as a sum of paths through the networks. We illustrate that network amplification (the bullwhip effect) can be effectively mitigated by reallocating flows within the network and provide a systematic methodology for doing so by minimizing the $\mathcal{H}_\infty$ norm of the network, illustrated using a 14 city example. Future work will extend these results to system with more general ordering policies, stochastic lead times, capacity limits and eventually to networks with cycles.

\appendix
\section{Proof of Theorem~\ref{thm:path}}
\label{app:path}
The proof is by induction on the topological index of A3. Fix
$\ell \in \Nc \setminus \Sc_0$ and set $v_i = 0$ for all $i \neq \ell$. If $j = \ell$
then $\Ppaths{\ell}{\ell} = \{(\ell)\}$ and
\begin{equation}
  T_{\ell \to \ell} = o_\ell/v_\ell = T_\ell
  = T_\ell\,\Tp{(\ell)}.
  \label{eq:path_base}
\end{equation}
Fix $j > \ell$ and suppose \eqref{eq:path_sum} holds for every
$\ell \leq k < j$. By \eqref{eq:Ti}, \eqref{eq:coupling}, and the
inductive hypothesis,
\begin{subequations}
\begin{align}
o_j = T_j d_j &= T_j \Bigl(\sum_{k \in \Cc(j)} w_{kj}\, o_k
        + v_j\Bigr)\\
        & = T_j \sum_{k \in \Cc(j)} w_{kj}\,
          T_{\ell\to k}\, v_\ell\\
       &= T_j \sum_{k \in \Cc(j)} \sum_{q \in \Ppaths{\ell}{k}}
       \Tp{(q,j)}\, v_\ell,\label{eq:prf1}
\end{align}
\end{subequations}
using $w_{kj}\, T_k\, \Tp{q} = \Tp{(q,j)}$. The double sum contains one term per path because appending $j$ to each path $q \in \Ppaths{\ell}{k}$ for all customers $k \in \Cc(j)$ constructs the complete set of paths $\Ppaths{\ell}{j}$, shown as the union
\begin{equation}
  \Ppaths{\ell}{j} = \cup_{k \in \Cc(j)}
  \cup_{q \in \Ppaths{\ell}{k}} \{(q, j)\},
  \label{eq:path_union}
\end{equation}
so the double sum in \eqref{eq:prf1} evaluates to
\begin{equation}
 \frac{o_j}{v_\ell} = T_{\ell\to j} = T_j \sum_{p \in \Ppaths{\ell}{j}}
       \Tp{p},
\end{equation}
completing the proof. \qed

\section{Proof of Theorem~\ref{thm:grad}}
\label{app:adjoint}

Equation \eqref{eq:chain} leaves the derivatives
$\partial\bar\sigma_k/\partial\theta$ to be computed. Where $\bar\sigma_k$ is
simple it is analytic in $\theta$, with
\citep{apkarian2006nonsmooth}
\begin{equation}
  \pd{\bar\sigma_k}{\theta} = \operatorname{Re}\Bigl[ \eta_k^{\mathsf H}\,
  \pd{T(\im\omega_k)}{\theta}\, \xi_k \Bigr].
  \label{eq:sigmader}
\end{equation}
Since $\xi_k$ does not depend on $\theta$, the derivative in \eqref{eq:sigmader}
appears only as $(\partial T/\partial\theta)\, \xi_k = \partial o/\partial\theta$,
the derivative of the orders $T(\im\omega_k)\, \xi_k$, so with \eqref{eq:chain}
\begin{equation}
  \pd{J_\rho}{\theta}
  = \sum_{\omega_k \in \Omega} \operatorname{Re}\Bigl[ \sum_{i}
    c_k\, \overline{\eta_{k,i}}\ \pd{o_i}{\theta} \Bigr],
  \label{eq:start}
\end{equation}
where $o = o^k$ is the order vector of Theorem~\ref{thm:grad}. Everything below is at one grid frequency $\omega_k$, evaluated at $\im\omega_k$ with the superscript $k$ of Theorem~\ref{thm:grad} dropped, and the sum over $k$ in \eqref{eq:start}
is the sum over frequencies in \eqref{eq:gw}--\eqref{eq:gk}. Every relation used
is complex-linear, so the real part is taken only at the end.

At node $i$, \eqref{eq:coupling} and \eqref{eq:Ti} give
\begin{equation}
\begin{gathered}
  \pd{d_j}{o_i} = w_{ij} \ \ (j \in \Sc(i)),
  \qquad
  \pd{o_i}{d_i} = T_i,\\
  \pd{d_j}{w_{ij}} = o_i .
\end{gathered}
  \label{eq:local}
\end{equation}
$o_i$ appears in \eqref{eq:start} with the factor $c_k \overline{\eta_{k,i}}$ and,
by the first of \eqref{eq:local}, in the demand of every supplier
$j \in \Sc(i)$; $d_i$ appears only in $o_i$. The chain rule therefore gives
\begin{equation}
  \pd{J_\rho}{o_i} = c_k\, \overline{\eta_{k,i}}
    + \sum_{j \in \Sc(i)} w_{ij}\, \pd{J_\rho}{d_j},
  \qquad
  \pd{J_\rho}{d_i} = T_i\, \pd{J_\rho}{o_i},
  \label{eq:chainnode}
\end{equation}
Writing $\psi_i = \partial J_\rho/\partial o_i$, the second relation gives $\partial J_\rho/\partial d_i = T_i\psi_i$, and substituting it into the first gives \eqref{eq:adjoint}; a source $j \in \Sc_0$ does not order, so $\partial J_\rho/\partial d_j = 0$ there. Each design variable changes the
cost through $o_i$, and $w_{ij}$ changes it a second time through $d_j$, which
contains the term $w_{ij} o_i$, so a second application of the chain rule
gives
\begin{equation}
\begin{aligned}
  \pd{J_\rho}{w_{ij}} &= \psi_{i}\, \pd{o_i}{w_{ij}} + T_j\,\psi_j\, o_i,\\
  \pd{J_\rho}{\tau_{ij}} &= \psi_{i}\, \pd{o_i}{\tau_{ij}},
  \qquad
  \pd{J_\rho}{K_i} = \psi_{i}\, \pd{o_i}{K_i} .
\end{aligned}
  \label{eq:assembled}
\end{equation}
Since $o_i = T_i d_i$, differentiating \eqref{eq:W_i} and \eqref{eq:Ti} gives
\begin{equation}
\begin{aligned}
  \pd{T_i}{W_i} &= -\frac{K_i T_i^{2}}{\im\omega_k + K_i},
  \qquad
  \pd{T_i}{K_i} = \frac{\im\omega_k (1 - W_i) T_i^{2}}
                           {(\im\omega_k + K_i)^{2}},\\
  \pd{W_i}{w_{ij}} &= (1 - \beta_i)\, e^{-\im\omega_k \tau_{ij}},\\
  \pd{W_i}{\tau_{ij}} &= -\im\omega_k (1 - \beta_i)\, w_{ij}\,
                          e^{-\im\omega_k \tau_{ij}} .
\end{aligned}
  \label{eq:elem}
\end{equation}
Substituting \eqref{eq:elem} into \eqref{eq:assembled} and using $o_i = T_i d_i$
to remove $d_i$, the contribution at $\omega_k$ is $\Lambda_i$ of
\eqref{eq:lambda} times the weight and delay derivatives of $W_i$, together with
the routing term $T_j \psi_j o_i$, which on taking real parts and summing over
$k$ gives \eqref{eq:gw}--\eqref{eq:gk}, the sign of the delay derivative
reversing with $\delta_{ij} = \tau^0_{ij} - \tau_{ij}$.

The derivative of the disk margin in the constraint \eqref{eq:feasible} follows
the same steps. With $\omega_i^\star$ the frequency attaining the supremum in
\eqref{eq:disk}, and everything evaluated at $s = \im\omega_i^\star$,
\begin{equation}
\begin{gathered}
  \pd{\alpha_i^{-1}}{\theta} = \operatorname{Re}
  \Bigl[ \frac{\overline{S_i - \tfrac12}}{|S_i - \tfrac12|}\,
  \pd{S_i}{\theta} \Bigr],
  \qquad
  \pd{S_i}{\theta} = -S_i^{2}\, \pd{L_i}{\theta},\\
  \pd{L_i}{w_{ij}} = \frac{K_i(1-\beta_i)}{s}\, e^{-s\tau_{ij}},
  \qquad
  \pd{L_i}{K_i} = \frac{W_i}{s},\\
  \pd{L_i}{\tau_{ij}} = -K_i(1-\beta_i)\, w_{ij}\, e^{-s\tau_{ij}} ,
\end{gathered}
  \label{eq:alphader}
\end{equation}
which involve only node $i$'s own weights, delays and gain. \qed

\bibliographystyle{ieeetr}
\bibliography{references}

@article{lee1997distortion,
  title={Information distortion in a supply chain: The bullwhip effect},
  author={Lee, Hau L and Padmanabhan, Venkata and Whang, Seungjin},
  journal={Management science},
  volume={43},
  number={4},
  pages={546--558},
  year={1997},
  publisher={Informs}
}

@article{wang2016progress,
  title={The bullwhip effect: Progress, trends and directions},
  author={Wang, Xun and Disney, Stephen M},
  journal={European Journal of Operational Research},
  volume={250},
  number={3},
  pages={691--701},
  year={2016},
  publisher={Elsevier}
}

@article{sterman1989beergame,
  title={Modeling managerial behavior: Misperceptions of feedback in a dynamic decision making experiment},
  author={Sterman, John D},
  journal={Management science},
  volume={35},
  number={3},
  pages={321--339},
  year={1989},
  publisher={INFORMS}
}

@article{metters1997quantifying,
  title={Quantifying the bullwhip effect in supply chains},
  author={Metters, Richard},
  journal={Journal of operations management},
  volume={15},
  number={2},
  pages={89--100},
  year={1997},
  publisher={Wiley Online Library}
}

@article{dejonckheere2003control,
  title={Measuring and avoiding the bullwhip effect: A control theoretic approach},
  author={Dejonckheere, Jeroen and Disney, Stephen M and Lambrecht, Marc R and Towill, Denis R},
  journal={European journal of operational research},
  volume={147},
  number={3},
  pages={567--590},
  year={2003},
  publisher={Elsevier}
}

@article{gaalman2006fullstate,
  title={Bullwhip reduction for ARMA demand: The proportional order-up-to policy versus the full-state-feedback policy},
  author={Gaalman, Gerard},
  journal={Automatica},
  volume={42},
  number={8},
  pages={1283--1290},
  year={2006},
  publisher={Elsevier}
}

@inproceedings{li2024mitigating,
  title={Mitigating transient bullwhip effects under imperfect demand forecasts},
  author={Li, Sarah HQ and D{\"o}rfler, Florian},
  booktitle={2024 IEEE 63rd Conference on Decision and Control (CDC)},
  pages={529--534},
  year={2024},
  organization={IEEE}
}

@article{melo2009facility,
  title={Facility location and supply chain management--A review},
  author={Melo, M Teresa and Nickel, Stefan and Saldanha-Da-Gama, Francisco},
  journal={European journal of operational research},
  volume={196},
  number={2},
  pages={401--412},
  year={2009},
  publisher={Elsevier}
}

@article{snyder2005reliability,
  title={Reliability models for facility location: the expected failure cost case},
  author={Snyder, Lawrence V and Daskin, Mark S},
  journal={Transportation science},
  volume={39},
  number={3},
  pages={400--416},
  year={2005},
  publisher={INFORMS}
}

@article{perera2017network,
  title={Network science approach to modelling the topology and robustness of supply chain networks: a review and perspective},
  author={Perera, Supun and Bell, Michael GH and Bliemer, Michiel CJ},
  journal={Applied network science},
  volume={2},
  number={1},
  pages={33},
  year={2017},
  publisher={Springer}
}

@inproceedings{ghosh2006growing,
  title={Growing well-connected graphs},
  author={Ghosh, Arpita and Boyd, Stephen},
  booktitle={Proceedings of the 45th IEEE Conference on Decision and Control},
  pages={6605--6611},
  year={2006},
  organization={IEEE}
}

@article{darmawan2025scnd,
  title={Supply chain network design with the presence of the bullwhip effect},
  author={Darmawan, Agus and Wong, Hartanto Wijaya and Fransoo, Jan},
  journal={International journal of production economics},
  volume={286},
  pages={109668},
  year={2025},
  publisher={Elsevier}
}

@article{farhat2021hinf,
  title={{$H_\infty$} network optimization for edge consensus},
  author={Farhat, Omar and Abou Jaoude, Dany and de Badyn, Mathias Hudoba},
  journal={European Journal of Control},
  volume={62},
  pages={2--13},
  year={2021},
  publisher={Elsevier}
}

@article{ghaedsharaf2019performance,
  title={Performance improvement in noisy linear consensus networks with time-delay},
  author={Ghaedsharaf, Yaser and Siami, Milad and Somarakis, Christoforos and Motee, Nader},
  journal={IEEE Transactions on automatic control},
  volume={64},
  number={6},
  pages={2457--2472},
  year={2018},
  publisher={IEEE}
}

@article{welikala2025consensus,
  title={Inventory consensus control in supply chain networks using dissipativity-based control and topology co-design},
  author={Welikala, Shirantha and Lin, Hai and Antsaklis, Panos J},
  journal={arXiv preprint arXiv:2502.06580},
  year={2025}
}

@article{ouyang2010bullwhip,
  title={The bullwhip effect in supply chain networks},
  author={Ouyang, Yanfeng and Li, Xiaopeng},
  journal={European Journal of Operational Research},
  volume={201},
  number={3},
  pages={799--810},
  year={2010},
  publisher={Elsevier}
}

@article{helbing2004physics,
  title={Physics, stability, and dynamics of supply networks},
  author={Helbing, Dirk and L{\"a}mmer, Stefan and Seidel, Thomas and {\v{S}}eba, P{\'e}tr and P{\l}atkowski, Tadeusz},
  journal={Physical Review E},
  volume={70},
  number={6},
  pages={066116},
  year={2004},
  publisher={APS}
}

@article{sipahi2009stability,
  title={On stability problems of supply networks constrained with transport delay},
  author={Sipahi, Rifat and L{\"a}mmer, Stefan and Helbing, Dirk and Niculescu, Silviu-Iulian},
  journal={ASME Journal of Dynamic Systems, Measurement, and Control},
  volume={131}, number={2}, pages={021005},
  year={2009}
}

@article{hernandezsantos2026stability,
  title={Stability Limits of Coordinated Supply Chains Under Transportation Delays: Implications for Resilient Logistics Design},
  author={Hernandez-Santos, Carlos and Martinez-Malacara, Gloria A and de la Cruz, Nain and Reynoso-Guajardo, Luis Alejandro and Hernandez-Vega, Jose Isidro and Gallardo-Morales, Mario Carlos and Macias-Tobias, Francisco Fabian and Hernandez, Amadeo and Garcia-Andrade, Roxana},
  journal={Systems},
  volume={14},
  number={7},
  pages={752},
  year={2026},
  publisher={MDPI}
}

@article{seiler2020introduction,
  title={An introduction to disk margins [lecture notes]},
  author={Seiler, Peter and Packard, Andrew and Gahinet, Pascal},
  journal={IEEE Control Systems Magazine},
  volume={40},
  number={5},
  pages={78--95},
  year={2020},
  publisher={IEEE}
}

@incollection{kreisselmeier1979systematic,
  title={Systematic control design by optimizing a vector performance index},
  author={Kreisselmeier, Gerhard and Steinhauser, Reinhold},
  booktitle={Computer aided design of control systems},
  pages={113--117},
  year={1980},
  publisher={Elsevier}
}

@article{apkarian2006nonsmooth,
  title={Nonsmooth $H_\infty$ synthesis},
  author={Apkarian, Pierre and Noll, Dominikus},
  journal={IEEE Transactions on Automatic Control},
  volume={51},
  number={1},
  pages={71--86},
  year={2006},
  publisher={IEEE}
}

@inproceedings{summers2015topology,
  title={Topology design for optimal network coherence},
  author={Summers, Tyler and Shames, Iman and Lygeros, John and D{\"o}rfler, Florian},
  booktitle={2015 European Control Conference (ECC)},
  pages={575--580},
  year={2015},
  organization={IEEE}
}

@article{hosoda2006altruistic,
  title={The governing dynamics of supply chains: The impact of altruistic behaviour},
  author={Hosoda, Takamichi and Disney, Stephen M},
  journal={Automatica},
  volume={42},
  number={8},
  pages={1301--1309},
  year={2006},
  publisher={Elsevier}
}

@article{qiu2015stabilization,
  title={Stabilization of supply networks with transportation delay and switching topology},
  author={Qiu, Xiang and Yu, Li and Zhang, Dan},
  journal={Neurocomputing},
  volume={155},
  pages={247--252},
  year={2015},
  publisher={Elsevier}
}

@article{fu2020cooperative,
  title={A cooperative distributed model predictive control approach to supply chain management},
  author={Fu, Dongfei and Zhang, Hai-Tao and Dutta, Abhishek and Chen, Guanrong},
  journal={IEEE Transactions on Systems, Man, and Cybernetics: Systems},
  volume={50},
  number={12},
  pages={4894--4904},
  year={2019},
  publisher={IEEE}
}

@article{dominguez2015structure,
  title={The impact of the supply chain structure on bullwhip effect},
  author={Dominguez, Roberto and Cannella, Salvatore and Framinan, Jose M},
  journal={Applied Mathematical Modelling},
  volume={39},
  number={23-24},
  pages={7309--7325},
  year={2015},
  publisher={Elsevier}
}

@article{disney2002discrete,
  title={A discrete transfer function model to determine the dynamic stability of a vendor managed inventory supply chain},
  author={Disney, Stephen Michael and Towill, Denis Royston},
  journal={International Journal of Production Research},
  volume={40},
  number={1},
  pages={179--204},
  year={2002},
  publisher={Taylor \& Francis}
}

@book{zhou1996robust,
  title={Robust and optimal control},
  author={Zhou, Kemin and Doyle, John Comstock and Glover, Keith and others},
  volume={40},
  year={1996},
  publisher={Prentice hall New Jersey}
}

@article{sarimveis2008dynamic,
  title={Dynamic modeling and control of supply chain systems: A review},
  author={Sarimveis, Haralambos and Patrinos, Panagiotis and Tarantilis, Chris D and Kiranoudis, Chris T},
  journal={Computers \& operations research},
  volume={35},
  number={11},
  pages={3530--3561},
  year={2008},
  publisher={Elsevier}
}

@article{ouyang2006characterization,
  title={Characterization of the bullwhip effect in linear, time-invariant supply chains: some formulae and tests},
  author={Ouyang, Yanfeng and Daganzo, Carlos},
  journal={Management science},
  volume={52},
  number={10},
  pages={1544--1556},
  year={2006},
  publisher={INFORMS}
}

@article{chen2000quantifying,
  title={Quantifying the bullwhip effect in a simple supply chain: The impact of forecasting, lead times, and information},
  author={Chen, Frank and Drezner, Zvi and Ryan, Jennifer K and Simchi-Levi, David},
  journal={Management science},
  volume={46},
  number={3},
  pages={436--443},
  year={2000},
  publisher={INFORMS}
}

@inproceedings{ignaciuk2020quantifying,
  title={Quantifying the bullwhip effect in networked structures with nontrivial topologies},
  author={Ignaciuk, Przemys{\l}aw and Dziomdziora, Adam},
  booktitle={Proceedings of the 2020 International Conference on Big Data in Management},
  pages={62--66},
  year={2020}
}

@article{zemzam2019stabilization,
  title={Stabilization of two-echelon supply networks with uncertain demand, multiple delays and switching topology using robust control},
  author={Zemzam, Azeddine and El Alami, Jamila and El Alami, Nourddine},
  journal={International Journal of Dynamics and Control},
  volume={7},
  number={1},
  pages={388--404},
  year={2019},
  publisher={Springer}
}

@article{wei2015exploring,
  title={Exploring the impact of network structure and demand collaboration on the dynamics of a supply chain network using a robust control approach},
  author={Wei, Yongchang and Wang, Hongwei and Chen, Fangyu},
  journal={Mathematical Problems in Engineering},
  volume={2015},
  number={1},
  pages={102727},
  year={2015},
  publisher={Wiley Online Library}
}

\end{document}